\documentclass[notitlepage,aps,prd,a4paper,11pt,amsfonts,amssymb,amsmath]{revtex4-1}
\usepackage{amsthm}
\usepackage[all]{xy}
\usepackage{enumerate}

\theoremstyle{plain}
\newtheorem{prop}{Proposition}
\theoremstyle{plain}
\newtheorem{cond}{Condition}
\theoremstyle{definition}
\newtheorem{defi}{Definition}

\newcounter{fcounter}

\newcounter{fpcounter}

\newcounter{ccounter}

\DeclareMathOperator{\tr}{tr}

\DeclareMathOperator{\sgn}{sgn}

\begin{document}
\title{Extensions of Lorentzian spacetime geometry:\\From Finsler to Cartan and vice versa}

\author{Manuel Hohmann}
\email{manuel.hohmann@ut.ee}
\affiliation{Teoreetilise F\"u\"usika Labor, F\"u\"usika Instituut, Tartu \"Ulikool, Riia 142, 51014 Tartu, Estonia}

\begin{abstract}
We briefly review two recently developed extensions of the Lorentzian geometry of spacetime and prove that they are in fact closely related. The first is the concept of observer space, which generalizes the space of Lorentzian observers, i.e., future unit timelike vectors, using Cartan geometry. The second is the concept of Finsler spacetimes, which generalizes the Lorentzian metric of general relativity to an observer-dependent Finsler metric. We show that every Finsler spacetime possesses a well-defined observer space that can naturally be equipped with a Cartan geometry. Conversely, we derive conditions under which a Cartan geometry on observer space gives rise to a Finsler spacetime. We further show that these two constructions complement each other. We finally apply our constructions to two gravity theories, MacDowell-Mansouri gravity on observer space and Finsler gravity, and translate their actions from one geometry to the other.
\end{abstract}
\maketitle

\section{Motivation}\label{sec:motivation}
General covariance is one of the most evident features of general relativity. It tells us that all physical quantities can be expressed by sections of some tensor bundle over the spacetime manifold~\(M\), or tensor fields for short. However, when we measure a physical quantity in practice, the measurement does not directly yield an element of some tensor space. Instead we always measure the components of a tensor with respect to a basis of the tensor space, which we must choose before performing the measurement and which is inbuilt into the measurement apparatus. A common way to specify a basis is by providing an orthonormal frame field, i.e., vector fields~\({(f_i, i = 0,\ldots,3)}\) which are orthonormal with respect to the Lorentzian metric of spacetime, \(g_{ab}f_i^af_j^b = \eta_{ij}\), and which yield a basis of the tangent spaces \(T_xM\) for all \(x \in M\). The tensorial nature of physical quantities provides a prescription how measurements performed using different frames are related.

Despite being a central building block of general relativity, the principle of general covariance is broken in various extensions of general relativity. In particular, physical quantities might possess a non-tensorial dependence on the choice of the timelike component \(f_0\) of a frame, which is conventionally identified with the four-velocity of the observer performing the measurement. A large class of such theories are based on the concept to split spacetime into space and time and focus of the dynamics of the geometry of space. Examples of such theories are given by geometrodynamic theories such as loop quantum gravity~\cite{Thiemann:2007zz,Ashtekar:1987gu} and sum-over-histories formulations such as spin foam models~\cite{Rovelli:1995ac,Reisenberger:1996pu,Barrett:1997gw,Baez:1997zt} and causal dynamical triangulations~\cite{Ambjorn:1998xu,Ambjorn:2001cv,Ambjorn:2005qt}. Moreover, a non-trivial dependence of physical quantities on \(f_0\) may be introduced by a breaking of local Lorentz invariance, for example, by a preferred field of observers, or test particles, described by a future unit timelike vector field~\cite{Brown:1994py,Jacobson:2000xp}.

The explicit dependence of physical quantities on the choice of an observer motivates the notion of observer space~\cite{Gielen:2012fz}. Instead of choosing a preferred observer, and thus a preferred four-velocity, one considers the space \(O\) of \emph{all} possible observers. In the case of a Lorentzian spacetime the observer space is simply given as the space of all future unit timelike vectors. This approach has several advantages: the observer-dependence of physical quantities becomes apparent by promoting them to functions on observer space, while a relativistic nature of spacetime is retained since no observer is preferred. Therefore, the concept of observer space unifies ideas from both covariant and canonical approaches to spacetime geometry.

Another key concept of general relativity states that the geometry of spacetime does not only provide a static background for physical theories, but is by itself a dynamical object which describes the dynamics of gravity. If we wish to lift this concept from spacetime to observer space, we need to equip observer space with a geometry that is governed by a suitable choice of gravitational dynamics. In this article we will choose a connection based approach and describe the geometry of observer space in terms of Cartan geometry~\cite{Cartan,Sharpe}. This choice allows not only a formulation of gravitational dynamics using MacDowell-Mansouri gravity~\cite{MacDowell:1977jt}, but also the reconstruction of the metric geometry of a Lorentzian spacetime from its observer space~\cite{Gielen:2012fz}.

The notion of observer space is not limited to the future unit timelike vectors of a Lorentzian manifold. In fact, every mathematical framework that is used to describe the concept of observers should provide a definition of observer space. The framework we discuss in this article is the description of spacetimes in terms of Finsler geometry~\cite{Chern,Bucataru}. Models of this type have been introduced as extensions to Einstein and string gravity~\cite{Horvath,Vacaru:2010fi,Vacaru:2007ng,Vacaru:2002kp}. In particular we will consider a recently developed definition of Finsler spacetimes, which provides a suitable background for both gravity and matter field theories such as electrodynamics, as well as a consistent causal structure extending the causal structure of Lorentzian spacetimes~\cite{Pfeifer:2011tk,Pfeifer:2011xi,Pfeifer:2011ve}. This causal structure and the notion of observers play a crucial role for the concept of Finsler spacetimes. Hence, the definition of a Finsler spacetime comprises the definition of an observer space.

It appears natural to ask whether one can describe the observer space of a Finsler spacetime using Cartan geometry, in the same way as it has been done for the future unit timelike vectors on a Lorentzian manifold. This would provide a simple and physically well-motivated example of a non-Lorentzian observer space Cartan geometry. Conversely, one may ask whether a given observer space Cartan geometry is induced by a Finsler spacetime, and whether this Finsler spacetime may be reconstructed. It is the central aim of this article to show that the former is indeed always possible, while the latter requires certain conditions to be satisfied.

The outline of this article is as follows. We will start with a brief review of the necessary mathematical foundations and their applications to spacetime physics in section~\ref{sec:foundations}. In particular, we will discuss Cartan geometry of observer space in section~\ref{subsec:cartan} and Finsler spacetimes in section~\ref{subsec:finsler}. We will then come to the main part of this article and show how these two different geometrical descriptions of spacetime can be connected. In section~\ref{sec:finslercartan} we will construct the observer space of a Finsler spacetime and equip it with a Cartan geometry. The opposite direction, constructing a Finsler geometry from a Cartan geometry on observer space, will be discussed in section~\ref{sec:cartanfinsler}. In section~\ref{sec:reconstruction} we will discuss whether we can reconstruct a given Finsler spacetime from its observer space Cartan geometry and vice versa. As a physical application of our constructions we will translate the actions of two gravity theories between the different geometric descriptions in section~\ref{sec:gravity}: MacDowell-Mansouri gravity on observer space in section~\ref{subsec:cartangrav} and Finsler gravity in section~\ref{subsec:finslergrav}. We will end with a conclusion in section~\ref{sec:conclusion}.

\section{Foundations}\label{sec:foundations}
The construction we present in this article is based on two different generalizations of the Lorentzian geometry of spacetime. The aim of this section is to briefly review these two generalizations. First we will discuss the concept of observer space and its geometric description in terms of Cartan geometry in section~\ref{subsec:cartan}. We will then discuss Finsler spacetimes in section~\ref{subsec:finsler} and review some basic concepts of Finsler geometry.

\subsection{Cartan geometry of observer space}\label{subsec:cartan}
In this section we give a brief summary of the Cartan geometric description of observer space presented in~\cite{Gielen:2012fz}. We start with a very basic review of Cartan geometry in general, before we apply this general framework to the geometries of spacetime and observer space. We briefly discuss the symmetry algebras of the maximally symmetric Lorentzian spacetimes, since they will form a crucial ingredient of our construction shown in the remainder of this article.

In order to describe the geometry of observer space, we employ a framework developed by Cartan under the name ``method of moving frames''~\cite{Cartan}. The basic idea of his construction is to locally compare the geometry of a manifold \(M\) to that of a homogeneous space, i.e., the coset space \(G/H\) of a Lie group \(G\) and a closed subgroup \(H \subset G\). Homogeneous spaces were extensively studied in Klein's Erlangen program and are hence also known as Klein geometries. Among other properties they carry the structure of a principal $H$-bundle \(\pi: G \to G/H\) and a connection given by the Maurer-Cartan 1-form \(A \in \Omega^1(G,\mathfrak{g})\) on \(G\) taking values in the Lie algebra \(\mathfrak{g}\) of \(G\). Using these structures in order to describe the local geometry of \(M\), we can define a Cartan geometry as follows:
\begin{defi}[Cartan geometry]
A \emph{Cartan geometry} modeled on a Klein geometry \(G/H\) is a principal $H$-bundle \(\pi: P \to M\) together with a $\mathfrak{g}$-valued 1-form \(A \in \Omega^1(P,\mathfrak{g})\) on \(P\), such that
\begin{list}{\textrm{C\arabic{ccounter}}.}{\usecounter{ccounter}}
\item\label{cartan:isomorphism}
For each \(p \in P\), \(A_p: T_pP \to \mathfrak{g}\) is a linear isomorphism.
\item\label{cartan:equivariant}
\(A\) is $H$-equivariant: \((R_h)^*A = \mathrm{Ad}(h^{-1}) \circ A\) \(\forall h \in H\).
\item\label{cartan:mcform}
\(A\) restricts to the Maurer-Cartan form on vertical vectors \(v \in \ker\pi_*\).
\end{list}
\end{defi}
Note in particular that condition~\ref{cartan:isomorphism} implies the existence of inverse functions \(\underline{A}_p: \mathfrak{g} \to T_pP\), corresponding to a map \(\underline{A}: \mathfrak{g} \to \Gamma(TP)\) which assigns to each element of the Lie algebra \(\mathfrak{g}\) a nowhere vanishing vector field, called its fundamental vector field. These, as well as further properties of Cartan geometry, will be discussed later in this article as they are needed. For now we restrict ourselves to a few practical examples and focus on the Cartan geometric description of spacetime and observer space. We start by discussing the model Klein geometries we will be working with.

Here and in the remainder of this article, let \(G\), \(H\) and \(K\) denote the frequently used groups
\begin{equation}\label{eqn:groups}
G = \begin{cases}
\mathrm{SO}_0(4,1) & \text{for } \Lambda > 0\\
\mathrm{ISO}_0(3,1) & \text{for } \Lambda = 0\\
\mathrm{SO}_0(3,2) & \text{for } \Lambda < 0
\end{cases}\,, \quad H = \mathrm{SO}_0(3,1)\,, \quad K = \mathrm{SO}(3)\,,
\end{equation}
where \(\mathrm{ISO}_0(3,1) = \mathrm{SO}_0(3,1) \ltimes \mathbb{R}^{3,1}\) is the proper orthochronous Poincar\'{e} group and the subscript~\(0\) indicates the connected component of the corresponding group. Note that \(G\) may denote any of the three groups listed above. The homogeneous spaces \(G/H\) associated to these groups are de Sitter space, Minkowski space and anti de Sitter space, corresponding to either a positive, zero or negative cosmological constant \(\Lambda\). The connection between the choice of the group \(G\) and the sign of the cosmological constant will be discussed in more detail in places where it becomes relevant. Whenever this is not noted explicitly, the result is independent of the choice of a particular group~\(G\).

Let \(\mathfrak{g}\), \(\mathfrak{h}\) and \(\mathfrak{k}\) denote the Lie algebras of \(G\), \(H\) and \(K\), respectively. For convenience we introduce a component notation for elements of these algebras. First observe that \(\mathfrak{g}\) splits into irreducible subrepresentations of the adjoint representation of \(H \subset G\),
\begin{equation}\label{eqn:algsplit1}
\mathfrak{g} = \mathfrak{h} \oplus \mathfrak{z}\,.
\end{equation}
These subspaces correspond to infinitesimal Lorentz transforms \(\mathfrak{h}\) and infinitesimal translations \(\mathfrak{z}\) of the homogeneous spacetimes \(G/H\). We can use this split to uniquely decompose any algebra element \(a \in \mathfrak{g}\) in the form
\begin{equation}\label{eqn:component}
a = h + z = \frac{1}{2}h^i{}_j\mathcal{H}_i{}^j + z^i\mathcal{Z}_i\,,
\end{equation}
where \(\mathcal{H}_i{}^j\) are the generators of \(\mathfrak{h} = \mathfrak{so}(3,1)\) and \(\mathcal{Z}_i\) are the generators of translations on \(G/H\). They satisfy the algebra relations
\begin{gather}
[\mathcal{H}_i{}^j,\mathcal{H}_k{}^l] = \delta^j_k\mathcal{H}_i{}^l - \delta^l_i\mathcal{H}_k{}^j + \eta_{ik}\eta^{lm}\mathcal{H}_m{}^j - \eta^{jl}\eta_{km}\mathcal{H}_i{}^m\,,\label{eqn:algebra}\\
[\mathcal{H}_i{}^j,\mathcal{Z}_k] = \delta^j_k\mathcal{Z}_i - \eta_{ik}\eta^{jl}\mathcal{Z}_l\,, \qquad [\mathcal{Z}_i,\mathcal{Z}_j] = \sgn\Lambda\,\eta_{ik}\mathcal{H}_j{}^k\,.\nonumber
\end{gather}
The last expression explicitly depends on the choice of the group \(G\), which can conveniently be expressed using the sign of the cosmological constant \(\Lambda\). In the following we will use this notation to describe the Cartan geometry of spacetime and observer space.

As a first example, let \((M,g)\) be a Lorentzian manifold and \(\tilde{\pi}: P \to M\) its oriented, time-oriented, orthonormal frame bundle. The proper orthochronous Lorentz group \(H = \mathrm{SO}_0(3,1)\) acts freely and transitively on the fibers of the frame bundle, hence turning it into a principal $H$-bundle. To promote this bundle to a Cartan geometry modeled on the Klein geometry \(G/H\), we need to construct a $\mathfrak{g}$-valued 1-form \(A\) on \(P\) that satisfies the conditions~\ref{cartan:isomorphism} to~\ref{cartan:mcform} listed above. We can divide this task into two parts by recalling that \(\mathfrak{g}\) splits into two subspaces according to~\eqref{eqn:algsplit1}. This induces a decomposition \(A = \omega + e\) of the Cartan connection into a $\mathfrak{h}$-valued part \(\omega\) and a $\mathfrak{z}$-valued part \(e\). The latter we set equal to the solder form, which in component notation can be written as
\begin{equation}\label{eqn:solderform}
e^i = f^{-1}{}^i_adx^a\,,
\end{equation}
where the coordinates \(f_i^a\) on the fibers of \(P\) are defined as the components of the frames \(f_i\) in the coordinate basis of the manifold coordinates \(x^a\), and \(f^{-1}{}^i_a\) denote the corresponding inverse frame components. For the $\mathfrak{h}$-valued part \(\omega\) we choose the Levi-Civita connection. Using the same component notation as above it reads
\begin{equation}\label{eqn:levicivita}
\omega^j{}_i = f^{-1}{}^j_adf_i^a + f^{-1}{}^j_af_i^b\Gamma^a{}_{bc}dx^c\,,
\end{equation}
where \(\Gamma^a{}_{bc}\) denotes the Christoffel symbols. It is not difficult to check that the $\mathfrak{g}$-valued 1-form~\(A\) defined above indeed satisfies conditions~\ref{cartan:isomorphism} to~\ref{cartan:mcform} of a Cartan connection, and thus defines a Cartan geometry modeled on \(G/H\). A well-known result of Cartan geometry states that the metric \(g\) can be reconstructed from the Cartan geometry, up to a global scale factor~\cite{Sharpe}.

We now turn our focus from the Cartan geometry of spacetime to the Cartan geometry of observer space. Let again \((M,g)\) be a Lorentzian manifold, \(\tilde{\pi}: P \to M\) its oriented, time-oriented, orthonormal frame bundle and \(O\) the space of all future unit timelike vectors on \(M\). Note that there exists a canonical projection \(\pi: P \to O\), which simply discards the spatial components of a frame. Since the rotation group \(K\) acts freely and transitively on the spatial frame components, and hence on the pre-images under \(\pi\), \(P \to O\) carries the structure of a principal $K$-bundle. We can promote this bundle to a Cartan geometry by choosing the same Cartan connection \(A\) on \(P\) as we already did in our first example displayed above. One easily checks that this choice again satisfies conditions~\ref{cartan:isomorphism} to~\ref{cartan:mcform}, but now for a Cartan geometry modeled on \(G/K\). Moreover, one can show that both the underlying spacetime manifold \(M\) and its Lorentzian metric \(g\) can be reconstructed from the observer space Cartan geometry, again up to a global scale factor~\cite{Gielen:2012fz}.

Since the Cartan geometric description of observer space is based on the model Klein geometry~\(G/K\) instead of \(G/H\), the subgroup \(K \subset G\) of spatial rotations takes a role similar to that of the Lorentz group \(H\) in the case of spacetime geometry. In particular, the adjoint representation of~\(K\) on the Lie algebra \(\mathfrak{g}\) induces a split of \(\mathfrak{g}\) in analogy to~\eqref{eqn:algsplit1} into its irreducible subrepresentations
\begin{equation}\label{eqn:algsplit2}
\mathfrak{g} = \mathfrak{k} \oplus \mathfrak{y} \oplus \vec{\mathfrak{z}} \oplus \mathfrak{z}^0\,.
\end{equation}
We can hence uniquely decompose every algebra element \(a \in \mathfrak{g}\) into these four components, which can be written as
\begin{equation}\label{eqn:decomposition}
a = \frac{1}{2}h^{\alpha}{}_{\beta}\mathcal{H}_{\alpha}{}^{\beta} + \frac{1}{2}(h^{\alpha}{}_0\mathcal{H}_{\alpha}{}^0 + h^0{}_{\alpha}\mathcal{H}_0{}^{\alpha}) + z^{\alpha}\mathcal{Z}_{\alpha} + z^0\mathcal{Z}_0\,,
\end{equation}
using the notation introduced in~\eqref{eqn:component} and with Greek indices taking values \(1,\ldots,3\). Cartan geometry provides a physical interpretation of these components as infinitesimal rotations, boosts and spatial and temporal translations.

This completes our brief review of the Cartan geometric description of spacetime and observer space. In the following section we discuss the notion of Finsler spacetimes, which form the second fundamental ingredient of our construction.

\subsection{Finsler spacetimes}\label{subsec:finsler}
In the previous section~\ref{subsec:cartan} we have reviewed the description of observer space in the language of Cartan geometry. Since it is our aim to construct an observer space geometry out of a Finsler geometry, we will now turn our focus to the Finsler geometric description of spacetimes. In this section we briefly review a definition of Finsler geometry which is suitable to equip a manifold with a Finsler metric of Lorentzian signature. We list the basic geometric objects on a Finsler spacetime and their properties, as far as they will become relevant for our construction in the following two sections.

The essential element of Finsler geometry is the length functional
\begin{equation}\label{eqn:finslerlength}
s[x] = \int d\tau\,F(x(\tau),\dot{x}(\tau))\,,
\end{equation}
which measures the length of a curve \(\tau \mapsto x(\tau)\). The Finsler function \(F: TM \to \mathbb{R}\) is a function on the tangent bundle \(TM\) of a manifold \(M\), on which we introduce coordinates as follows. Given coordinates \((x^a)\) on \(M\), we can denote any tangent vector by
\begin{equation}
y^a\left.\frac{\partial}{\partial x^a}\right|_x
\end{equation}
in the coordinate basis \(\partial/\partial x^a\) of \(TM\). The coordinates \((x^a,y^a)\) on \(TM\) are called induced coordinates and will be used throughout this article. For convenience we further introduce the notation
\begin{equation}
\left\{\partial_a = \frac{\partial}{\partial x^a}, \bar{\partial}_a = \frac{\partial}{\partial y^a}\right\}
\end{equation}
for the induced coordinate basis of \(TTM\).

In order to describe the geometry of spacetimes we need a generalization of the widely used notion of Finsler geometries with Euclidean signature~\cite{Chern,Bucataru} to Lorentzian signature. We will adapt the following definition of Finsler spacetimes, which is physically well-motivated since it provides a suitable causality and background geometry for electrodynamics and gravity~\cite{Pfeifer:2011tk,Pfeifer:2011xi}:
\begin{defi}[Finsler spacetime]
A \emph{Finsler spacetime} \((M,L,F)\) is a four-dimensional, connected, Hausdorff, paracompact, smooth manifold \(M\) equipped with continuous real functions \(L, F\) on the tangent bundle \(TM\) which has the following properties:
\begin{list}{\textrm{F\arabic{fcounter}}.}{\usecounter{fcounter}}
\item\label{finsler:lsmooth}
\(L\) is smooth on the tangent bundle without the zero section \(TM \setminus \{0\}\).
\item\label{finsler:lhomogeneous}
\(L\) is positively homogeneous of real degree \(n \geq 2\) with respect to the fiber coordinates of~\(TM\),
\begin{equation}
L(x,\lambda y) = \lambda^nL(x,y) \quad \forall \lambda > 0\,,
\end{equation}
and defines the Finsler function \(F\) via \(F(x,y) = |L(x,y)|^{\frac{1}{n}}\).
\item\label{finsler:lreversible}
\(L\) is reversible: \(|L(x,-y)| = |L(x,y)|\).
\item\label{finsler:lhessian}
The Hessian
\begin{equation}
g^L_{ab}(x,y) = \frac{1}{2}\bar{\partial}_a\bar{\partial}_bL(x,y)
\end{equation}
of \(L\) with respect to the fiber coordinates is non-degenerate on \(TM \setminus X\), where \(X \subset TM\) has measure zero and does not contain the null set \(\{(x,y) \in TM | L(x,y) = 0\}\).
\item\label{finsler:timelike}
The unit timelike condition holds, i.e., for all \(x \in M\) the set
\begin{equation}
\Omega_x = \left\{y \in T_xM \left| |L(x,y)| = 1, g^L_{ab}(x,y) \text{ has signature } (\epsilon,-\epsilon,-\epsilon,-\epsilon), \epsilon = \frac{L(x,y)}{|L(x,y)|}\right.\right\}
\end{equation}
contains a non-empty closed connected component~\({S_x \subseteq \Omega_x \subset T_xM}\).
\end{list}
\end{defi}
In contrast to the standard notion of Euclidean Finsler geometry, which is defined only by the Finsler function \(F\), this definition uses an auxiliary function \(L\) in order to define a suitable null structure via condition~\ref{finsler:lhessian} and timelike vectors via condition~\ref{finsler:timelike}. It further ensures that the Finsler function \(F\) has the following properties, as proven in~\cite{Pfeifer:2011tk}:
\begin{prop}
The Finsler function \(F\) of a Finsler spacetime \((M,L,F)\) satisfies:
\begin{list}{\textrm{F\arabic{fpcounter}}'.}{\usecounter{fpcounter}}
\item\label{finsler:fsmooth}
\(F\) is smooth on the tangent bundle where \(L\) is non-vanishing, \(TM \setminus \{L = 0\}\).
\item\label{finsler:fpositive}
\(F\) is non-negative, \(F(x,y) \geq 0\).
\item\label{finsler:fhomorev}
\(F\) is positively homogeneous of degree one in the fiber coordinates and reversible, i.e.,
\begin{equation}
F(x,\lambda y) = |\lambda|F(x,y) \quad \forall \lambda \in \mathbb{R}\,.
\end{equation}
\item\label{finsler:fmetric}
The Finsler metric
\begin{equation}
g^F_{ab}(x,y) = \frac{1}{2}\bar{\partial}_a\bar{\partial}_bF^2(x,y)
\end{equation}
of \(F\) is defined on \(TM \setminus \{L = 0\}\) and is non-degenerate on \(TM \setminus (X \cup \{L = 0\})\).
\end{list}
\end{prop}
In the following, however, we will not make use of the auxiliary function \(L\) since we will not be dealing with the null structure of Finsler spacetimes, and thus shorten the notation \((M,L,F)\) to~\((M,F)\). For our purposes it will be sufficient to work with the Finsler function \(F\) and its Finsler metric \(g^F\) only. In particular, they can be used to define the notion of Finsler geodesics, i.e., curves~\({\tau \mapsto x(\tau)}\) that minimize the length functional~\eqref{eqn:finslerlength}. The corresponding geodesic equation can be written in arc length parametrization as
\begin{equation}\label{eqn:geodesic}
\ddot{x}^a + N^a{}_b(x,\dot{x})\dot{x}^b = 0\,,
\end{equation}
where \(N^a{}_b\) denotes the coefficients of the the \emph{Cartan non-linear connection}
\begin{equation}\label{eqn:nonlinconn}
N^a{}_b = \frac{1}{4}\bar{\partial}_b\left[g^{F\,ac}(y^d\partial_d\bar{\partial}_cF^2 - \partial_cF^2)\right]\,.
\end{equation}
The Cartan non-linear connection induces a unique split of the tangent bundle \(TTM\) into horizontal and vertical parts, \(TTM = HTM \oplus VTM\). The horizontal tangent bundle \(HTM\) is spanned by the vector fields
\begin{equation}
\{\delta_a = \partial_a - N^b{}_a\bar{\partial}_b\}\,,
\end{equation}
while the vertical tangent bundle \(VTM\) is spanned by \(\{\bar{\partial}_a\}\). The corresponding basis \(\{\delta_a,\bar{\partial}_a\}\) that respects this split is called the \emph{Berwald basis}. Its dual basis of \(T^*TM\) is given by
\begin{equation}
\{dx^a, \delta y^a = dy^a + N^a{}_bdx^b\}\,.
\end{equation}
We finally remark that the tangent bundle \(TM\) of a Finsler spacetime is equipped with a metric~\(G_S\) called the \emph{Sasaki metric}, which can most conveniently be expressed in the Berwald basis as
\begin{equation}\label{eqn:sasakimetric}
G_S = g^F_{ab}dx^a \otimes dx^b + g^F_{ab}\delta y^a \otimes \delta y^b\,.
\end{equation}
One easily reads off that the horizontal and vertical tangent spaces are mutually orthogonal with respect to the Sasaki metric.

This completes our brief review of mathematical foundations. We now have all necessary tools at hand to show how Cartan geometry on observer space as displayed in the preceding section~\ref{subsec:cartan} and Finsler spacetimes as displayed in this section are related. This will be done in the following sections.

\section{Cartan geometry of Finsler observer space}\label{sec:finslercartan}
Using the foundations reviewed in the preceding section~\ref{sec:foundations} we are now in the position to construct the observer space of a Finsler spacetime and to equip it with a Cartan geometry modeled on the Klein geometry \(G/K\), where \(G\) and \(K\) denote the groups displayed in~\eqref{eqn:groups}. Our construction will proceed in several steps. In section~\ref{subsec:observerspace} we will construct the observer space of a Finsler spacetime and discuss some of its geometric properties. We will then choose a principal $K$-bundle over observer space and equip it with a Cartan connection in section~\ref{subsec:cartanconnection}. The fundamental vector fields associated to the Cartan connection will be constructed in section~\ref{subsec:fundvecfields}. We will then discuss various properties of this Cartan geometry. In section~\ref{subsec:tangbundsplit} we will show how the decomposition of \(\mathfrak{g}\) into irreducible subrepresentations of the adjoint representation of \(K\) leads to a split of the tangent bundle of observer space into four subbundles, each of which has a clear physical interpretation. One of these subbundles, which corresponds to the notion of time translation, will be discussed in further detail in section~\ref{subsec:timetranslation}. We finally calculate the curvature of the Cartan connection in section~\ref{subsec:curvature}.

\subsection{Observer space of Finsler spacetimes}\label{subsec:observerspace}
Our aim is to proceed in analogy to the case of a Lorentzian spacetime and to define the observer space \(O \subset TM\) of a Finsler spacetime as the space of future unit timelike vectors, which correspond to the four-velocities of physical observers. It follows from the causal structure of Finsler spacetimes~\cite{Pfeifer:2011tk} that at each point \(x \in M\) the physical observers are given by the set \(S_x\) from condition~\ref{finsler:timelike} shown in section~\ref{subsec:finsler}. We therefore define the observer space
\begin{equation}\label{eqn:observerspace}
O = \bigcup_{x \in M}S_x\,.
\end{equation}
Note that condition~\ref{finsler:timelike} ensures only existence but not uniqueness of \(S_x\). In fact, it follows from the reversibility property~\ref{finsler:lreversible} that if \(S_x\) is a closed connected component of \(\Omega_x\), also \(-S_x\) is a closed connected component. The choice of one of these components corresponds to a choice of time orientation on the Finsler spacetime~\cite{Pfeifer:2011tk}.

Before we proceed, we remark a few properties of Finsler observer space. It immediately follows from condition~\ref{finsler:timelike} that \(F(x,y) = 1\) for all \((x,y) \in O\). From properties~\ref{finsler:fhomorev} and \ref{finsler:fmetric} of the Finsler function it then follows that \(g^F_{ab}(x,y)y^ay^b = 1\). Note further that \(O\) is a seven-dimensional submanifold of \(TM\) and that its tangent spaces \(T_{(x,y)}O\) are precisely given by the tangent vectors~\({v \in T_{(x,y)}TM}\) which satisfy \(vF = 0\).

The tangent bundle \(TO\) of observer space deserves a deeper discussion. It naturally splits into a direct sum of three subbundles
\begin{equation}\label{eqn:tangbundsplit}
TO = VO \oplus HO = VO \oplus \vec{H}O \oplus H^0O\,,
\end{equation}
which are mutually orthogonal with respect to the restriction of the Sasaki metric~\eqref{eqn:sasakimetric} to \(O\). The vertical bundle \(VO\) is given by the kernel of the pushforward \(\pi'_*\) of the natural projection~\({\pi': O \to M, (x,y) \mapsto x}\) of observer space onto the underlying spacetime manifold. It is constituted by the tangent spaces to the three-dimensional submanifolds \(S_x \subset O\), and is thus a three-dimensional subbundle of \(TO\). The horizontal bundle \(HO\) is the four-dimensional orthogonal complement to \(VO\) and further splits into a three-dimensional spatial part \(\vec{H}O\) and a one-dimensional temporal part \(H^0O\). This split is induced by the contact form
\begin{equation}\label{eqn:contactform}
\alpha = g^F_{ab}y^adx^b \in \Omega^1(O)\,,
\end{equation}
which vanishes on the vertical bundle \(VO\). Its kernel on the horizontal bundle defines the spatial subbundle \(\vec{H}O\). Finally, the temporal subbundle \(H^0O\) is the orthogonal complement of \(\vec{H}O\) in~\(HO\). It is the span of the Reeb vector field
\begin{equation}\label{eqn:reebvectorfield}
\mathbf{r} = y^a\delta_a \in \Gamma(TO)\,,
\end{equation}
which is the unique vector field on \(O\) normalized by \(\alpha\) and whose flow preserves \(\alpha\), i.e., \(\alpha(\mathbf{r}) = 1\) and \(\mathcal{L}_{\mathbf{r}}\alpha = 0\), where \(\mathcal{L}\) denotes the Lie derivative. The split of the tangent bundle has a simple physical interpretation. The vertical tangent vectors correspond to infinitesimal boosts, which change the four-velocity \(y^a\) of an observer, but not its location \(x^a\). The horizontal tangent vectors correspond to infinitesimal translations changing the location of an observer without changing its four-velocity. The horizontal vectors further split into a spatial and a temporal part, corresponding to translations orthogonal or parallel to the observer four-velocity. We will re-encounter this split and its physical interpretation later in section~\ref{subsec:tangbundsplit}.

The one-dimensional temporal horizontal bundle \(H^0O\) spanned by the Reeb vector field~\eqref{eqn:reebvectorfield} is of particular interest. Curves on \(O\) whose tangent vectors lie in \(H^0O\) connect observers which differ by a time translation only. This in particular applies to integral curves \(\tau \mapsto (x(\tau),y(\tau))\) of the Reeb vector field \(\mathbf{r}\), which are characterized by two properties. First, they are canonical lifts of curves on the spacetime manifold \(M\), i.e., the tangent vectors of the projected curves \(\tau \mapsto x(\tau)\) on \(M\) satisfy \(\dot{x}(\tau) = y(\tau)\). Since \(y(\tau) \in O\) is unit timelike, it then immediately follows that the projected curves are timelike curves which are parametrized by their arc length. Curves of this type can be interpreted as observer trajectories. Second, the projected curves on \(M\) are Finsler geodesics, i.e., they satisfy the geodesic equation~\eqref{eqn:geodesic}. They therefore correspond to the time evolution of freely falling observers. This time evolution will be discussed further in section~\ref{subsec:timetranslation}.

\subsection{The Cartan connection}\label{subsec:cartanconnection}
In the previous section we have seen that the observer space \(O\) of a Finsler spacetime \((M,F)\) is a seven-dimensional submanifold of \(TM\), which is naturally equipped with several geometric structures. We now aim to promote \(O\) to a Cartan geometry modeled on \(G/K\). For this purpose we first need to construct a principal $K$-bundle \(\pi: P \to O\). We will then equip \(P\) with a Cartan connection \(A \in \Omega^1(P,\mathfrak{g})\), which satisfies the conditions~\ref{cartan:isomorphism} to~\ref{cartan:mcform} listed in section~\ref{subsec:cartan}.

Recall from section~\ref{subsec:cartan} that in the Lorentzian case we have chosen \(P\) to be the orthonormal, oriented, time-oriented frames on \(M\) and \(\pi\) the canonical projection to \(O\) that discards the spatial components of a frame. We can proceed in full analogy here and define \(P\) as the set of all oriented, time-oriented frames \((x,f)\) on \(M\) whose time component lies in \(O\) and which are orthonormal with respect to the Finsler metric, \(g^F_{ab}(x,f_0)f_i^af_j^b = -\eta_{ij}\). One easily checks that \(\pi: P \to O\) is indeed a principal $K$-bundle and that \(K\) acts on \(P\) by a spatial rotation of the frames. Note that in contrast to the Lorentzian case, \(P\) will in general not carry the structure of a principal $H$-bundle over \(M\).

We now equip this principal $K$-bundle with a Cartan connection \(A \in \Omega^1(P,\mathfrak{g})\), hence turning it into a Cartan geometry modeled on \(G/K\). Recall that the Lie algebra \(\mathfrak{g} = \mathfrak{h} \oplus \mathfrak{z}\) splits into a direct sum of vector spaces. Thus, also the Cartan connection \(A = \omega + e\) splits into two components~\({\omega \in \Omega^1(P,\mathfrak{h})}\) and \(e \in \Omega^1(P,\mathfrak{z})\). Our aim is again to proceed in analogy to the Lorentzian case, where \(e\) is given by the solder form and \(\omega\) is given by the Levi-Civita connection. In fact, the solder form can naturally be generalized to Finsler spacetimes and in our component notation be written as
\begin{equation}\label{eqn:connectionz}
e^i = f^{-1}{}^i_adx^a\,,
\end{equation}
which is identical to the corresponding expression~\eqref{eqn:solderform} on a Lorentzian manifold. We are thus left with constructing the $\mathfrak{h}$-valued part \(\omega\) of the Cartan connection, which generalizes the Levi-Civita connection~\eqref{eqn:levicivita}. Note that there is a simple geometric interpretation for \(\omega\) on a Lorentzian manifold. For a curve \(\tau \mapsto (x(\tau),f(\tau))\) on \(P\) it measures the covariant derivative of the frame vectors \(f_i\) along the projected curve \(\tau \mapsto x(\tau)\) on \(M\). For a tangent vector \(v \in TP\) this yields
\begin{equation}\label{eqn:covderivative}
f_j\omega^j{}_i(v) = \nabla_{\tilde{\pi}_*(v)}f_i\,.
\end{equation}
If we aim to generalize this geometric property to Finsler spacetimes, we find that in Finsler geometry covariant derivatives are not defined on the spacetime manifold \(M\), but on the tangent bundle \(TM\), and thus also on the submanifold \(O \subset TM\). It therefore appears natural to simply replace the projection \(\tilde{\pi}\) to the spacetime manifold \(M\) in~\eqref{eqn:covderivative} with the projection \(\pi\) to observer space \(O\), and to define the $\mathfrak{h}$-valued part \(\omega\) of the Cartan connection by
\begin{equation}\label{eqn:covderivative2}
f_j\omega^j{}_i(v) = \nabla_{\pi_*(v)}f_i\,.
\end{equation}
We finally need to find a suitable covariant derivative \(\nabla\) on the observer space \(O\), i.e., a suitable connection on the tangent bundle \(TM\). Since we are dealing with orthonormal frames, the connection we choose must in particular preserve this orthonormality property. We therefore need a connection on \(TM\) which is compatible with the Finsler metric. A connection that satisfies this condition is given by the Cartan linear connection~\cite{Cartan2}, which can be specified in terms of its covariant derivative
\begin{equation}
\nabla_{\delta_a}\delta_b = F^c{}_{ab}\delta_c\,, \quad \nabla_{\delta_a}\bar{\partial}_b = F^c{}_{ab}\bar{\partial}_c\,, \quad \nabla_{\bar{\partial}_a}\delta_b = C^c{}_{ab}\delta_c\,, \quad \nabla_{\bar{\partial}_a}\bar{\partial}_b = C^c{}_{ab}\bar{\partial}_c
\end{equation}
in the Berwald basis, where the connection coefficients \(F^a{}_{bc}\) and \(C^a{}_{bc}\) are given by
\begin{subequations}
\begin{align}
F^a{}_{bc} &= \frac{1}{2}g^{F\,ad}(\delta_bg^F_{cd} + \delta_cg^F_{bd} - \delta_dg^F_{bc})\,,\\
C^a{}_{bc} &= \frac{1}{2}g^{F\,ad}(\bar{\partial}_bg^F_{cd} + \bar{\partial}_cg^F_{bd} - \bar{\partial}_dg^F_{bc})\,.
\end{align}
\end{subequations}
Inserting this into~\eqref{eqn:covderivative2} we can solve for \(\omega^j{}_i\) and obtain
\begin{align}
\omega^j{}_i &= f^{-1}{}^j_adf_i^a + f^{-1}{}^j_af_i^b\left[F^a{}_{bc}dx^c + C^a{}_{bc}(N^c{}_ddx^d + df_0^c)\right]\nonumber\\
&= \frac{1}{2}\left(\delta_i^k\delta_l^j - \eta^{jk}\eta_{il}\right)f^{-1}{}^l_a\,df_k^a + \frac{1}{2}\eta^{jk}f_i^bf_k^c(\delta_bg^F_{ac} - \delta_cg^F_{ab})dx^a\,.\label{eqn:connectionh}
\end{align}
Together with~\eqref{eqn:connectionz} this determines a $\mathfrak{g}$-valued 1-form \(A\) on \(P\).

We still need to show that \(A\) satisfies the conditions on a Cartan connection listed in section~\ref{subsec:cartan}. A quick and straightforward calculation, which we omit here for brevity, shows that conditions~\ref{cartan:equivariant} and~\ref{cartan:mcform} are indeed satisfied. The proof of condition~\ref{cartan:isomorphism} is less obvious. We therefore postpone this proof to the next section, where we will discuss the fundamental vector fields associated to the Cartan connection. Our construction of the fundamental vector fields will then yield the desired proof as a corollary.

\subsection{Fundamental vector fields}\label{subsec:fundvecfields}
In the preceding section we have constructed a $\mathfrak{g}$-valued 1-form \(A\) on Finsler observer space, which satisfies conditions~\ref{cartan:equivariant} and~\ref{cartan:mcform} of a Cartan connection. We will now show that it also satisfies condition~\ref{cartan:isomorphism}, i.e., that the restriction \(A_{(x,f)}: T_{(x,f)}P \to \mathfrak{g}\) is a linear isomorphism for all~\({(x,f) \in P}\). For this purpose we will explicitly construct the inverse maps \(\underline{A}_{(x,f)}: \mathfrak{g} \to T_{(x,f)}P\). This corresponds to specifying a map \(\underline{A}: \mathfrak{g} \to \Gamma(TP)\) that assigns to each element of the Lie algebra \(\mathfrak{g}\) a vector field on \(P\). These vector fields are called the fundamental vector fields of the Cartan connection \(A\). For \(h \in \mathfrak{h}\) and \(z \in \mathfrak{z}\) they are the unique vector fields that satisfy
\begin{subequations}
\begin{align}
\omega(\underline{A}(h)) &= h\,, & e(\underline{A}(h)) &= 0\,,\\
\omega(\underline{A}(z)) &= 0\,, & e(\underline{A}(z)) &= z\,.
\end{align}
\end{subequations}
We can solve these equations for \(\underline{A}(h)\) and \(\underline{A}(z)\). Using the expressions~\eqref{eqn:connectionz} and~\eqref{eqn:connectionh} for \(e\) and \(\omega\) we find the component expressions
\begin{subequations}\label{eqn:fundvecfields}
\begin{align}
\underline{A}(h) &= \left(h^i{}_jf_i^a - h^i{}_0f^b_if^c_jC^a{}_{bc}\right)\bar{\partial}^j_a\,,\label{eqn:vecfieldh}\\
\underline{A}(z) &= z^if_i^a\left(\partial_a - f_j^bF^c{}_{ab}\bar{\partial}^j_c\right)\,,\label{eqn:vecfieldz}
\end{align}
\end{subequations}
where we introduced the notation
\begin{equation}
\partial_a = \frac{\partial}{\partial x^a}\,, \quad \bar{\partial}^i_a = \frac{\partial}{\partial f_i^a}
\end{equation}
for tangent vectors to the frame bundle \(P\). This proves that \(A_{(x,f)}\) possesses an inverse, in the sense that \(A_{(x,f)} \circ \underline{A}_{(x,f)} = \mathrm{id}_{\mathfrak{g}}\). Since \(T_{(x,f)}P\) and \(\mathfrak{g}\) are finite dimensional vector spaces of equal dimension, this is already sufficient to show that both \(A_{(x,f)}\) and \(\underline{A}_{(x,f)}\) are linear isomorphisms and that one is the inverse of the other. In thus follows that \(A\) indeed satisfies condition~\ref{cartan:isomorphism} and hence is a Cartan connection.

The map \(\underline{A}: \mathfrak{g} \to \Gamma(TP)\) assigns to each subspace of \(\mathfrak{g}\) a subbundle of \(TP\). In the following section we will apply this property to those subspaces of \(\mathfrak{g}\) which are invariant under the adjoint representation of \(K\) and discuss the associated subbundles.

\subsection{Split of the tangent bundle}\label{subsec:tangbundsplit}
Using the results from the previous two sections we now have an isomorphic mapping of each tangent space \(T_{(x,f)}P\) to the Lie algebra \(\mathfrak{g}\). In this section we make use of this mapping and carry the structure of \(\mathfrak{g}\) over to \(TP\). In particular, we consider the decomposition~\eqref{eqn:algsplit2} of the adjoint representation of \(K \subset G\) on \(\mathfrak{g}\) into a direct sum of irreducible subrepresentations. It induces a split of the Cartan connection \(A \in \Omega^1(P,\mathfrak{g})\) into 1-forms \(\Omega \in \Omega^1(P,\mathfrak{k})\), \(b \in \Omega^1(P,\mathfrak{y})\), \(\vec{e} \in \Omega^1(P,\vec{\mathfrak{z}})\) and \(e^0 \in \Omega^1(P,\mathfrak{z}^0)\). Each of these components isomorphically maps a subspace of \(T_{(x,f)}P\) to the corresponding subspace of \(\mathfrak{g}\), as shown in the following diagram:
\begin{equation}\label{eqn:tangbundsplit2}
\xymatrix{R_{(x,f)}P \ar@{^(->>}[d]_{\Omega} \ar@{}[r]|{\oplus} \ar@{}[dr]|{+} & B_{(x,f)}P \ar@{^(->>}[d]_b \ar@{}[r]|{\oplus} \ar@{}[dr]|{+} & \vec{H}_{(x,f)}P \ar@{^(->>}[d]_{\vec{e}} \ar@{}[r]|{\oplus} \ar@{}[dr]|{+} & H^0_{(x,f)}P \ar@{^(->>}[d]_{e^0} \ar@{}[r]|{=} \ar@{}[dr]|{=} & T_{(x,f)}P \ar@{^(->>}[d]_A\\
\mathfrak{k} \ar@{}[r]|{\oplus} & \mathfrak{y} \ar@{}[r]|{\oplus} & \vec{\mathfrak{z}} \ar@{}[r]|{\oplus} & \mathfrak{z}^0 \ar@{}[r]|{=} & \mathfrak{g}}
\end{equation}
This induces a decomposition of the tangent bundle \(TP\) into four subbundles, which are spanned by the fundamental vector fields \(\underline{A}(\mathfrak{k})\), \(\underline{A}(\mathfrak{y})\), \(\underline{A}(\vec{\mathfrak{z}})\) and \(\underline{A}(\mathfrak{z}^0)\). These subbundles have a simple geometric interpretation: the elements of \(RP\), \(BP\), \(\vec{H}P\) and \(H^0P\) correspond to infinitesimal rotations, boosts and spatial and temporal translations of an observer frame \((x,f)\).

Recall from section~\ref{subsec:observerspace} that the tangent bundle \(TO\) of observer space splits in a similar fashion into boosts and translations. Indeed the splits~\eqref{eqn:tangbundsplit} and~\eqref{eqn:tangbundsplit2} of \(TO\) and \(TP\) are closely related. This can be seen by calculating the pushforward of the fundamental vector fields, evaluated at~\({(x,f) \in P}\), along \(\pi\). For \(h \in \mathfrak{h}\), \(z \in \mathfrak{z}\) one obtains
\begin{equation}\label{eqn:vectorpush}
\pi_*(\underline{A}(h)) = h^{\alpha}{}_0f_{\alpha}^a\bar{\partial}_a\,, \qquad \pi_*(\underline{A}(z)) = z^if_i^a\delta_a\,.
\end{equation}
Comparing this with the decomposition~\eqref{eqn:decomposition} one can immediately see that \(\pi_* \circ \underline{A}\) vanishes on the subspace \(\mathfrak{k}\) of rotations, where only the spatial components \(h^{\alpha}{}_{\beta}\) are non-zero. Thus, \(\pi_*\) vanishes on the rotation subbundle \(RP\) spanned by the vector fields \(\underline{A}(\mathfrak{k})\). Similarly, one can read off that \(\pi_*\) isomorphically maps each subspace \(B_{(x,f)}P\) to \(V_{(x,f_0)}O\), and analogously for the spaces constituting the subbundles \(HP\) and \(HO\). This can conveniently be summarized using the following diagram:
\begin{equation}
\xymatrix{R_{(x,f)}P \ar[d]_{\pi_*} \ar@{}[r]|{\oplus} & B_{(x,f)}P \ar@{^(->>}[d]_{\pi_*} \ar@{}[r]|{\oplus} & \vec{H}_{(x,f)}P \ar@{^(->>}[d]_{\pi_*} \ar@{}[r]|{\oplus} & H^0_{(x,f)}P \ar@{^(->>}[d]_{\pi_*} \ar@{}[r]|{=} & T_{(x,f)}P \ar[d]_{\pi_*}\\
0 & V_{(x,f_0)}O \ar@{}[r]|{\oplus} & \vec{H}_{(x,f_0)}O \ar@{}[r]|{\oplus} & H^0_{(x,f_0)}O \ar@{}[r]|{=} & T_{(x,f_0)}O}
\end{equation}
The one-dimensional subbundle \(H^0P\), whose elements are mapped into the one-dimensional subbundle \(H^0O\), is of particular interest. We have seen already in section~\ref{subsec:observerspace} that \(H^0O\) is spanned by a vector field \(\mathbf{r}\) whose integral curves describe freely falling observers. In the next section we will show that \(H^0P\) plays a similar role in the Cartan geometric description of observer space.

\subsection{Time translation}\label{subsec:timetranslation}
In the previous section we have shown that the tangent bundle \(TP\) of \(P\) decomposes into several subbundles, and that this split is closely related to the split of the tangent bundle \(TO\) shown in section~\ref{subsec:observerspace}. Here we turn our attention to the one-dimensional subbundle \(H^0P\). It is spanned by the fundamental vector field
\begin{equation}\label{eqn:timevectorfield}
\mathbf{t} = \underline{A}(\mathcal{Z}_0) = f_0^a\partial_a - f_j^aN^b{}_a\bar{\partial}^j_b
\end{equation}
associated to the generator \(\mathcal{Z}_0\) of time translation. It is the unique vector field that satisfies the conditions \(e^i(\mathbf{t}) = \delta^i_0\) and \(\omega^i{}_j(\mathbf{t}) = 0\). The former condition, together with~\eqref{eqn:connectionz}, implies that an integral curve \(\tau \mapsto (x(\tau), f(\tau))\) of \(\mathbf{t}\) satisfies
\begin{equation}\label{eqn:canonicallift}
\dot{x}^a(\tau) = f_0^a(\tau)\,,
\end{equation}
i.e., the tangent vector \(\dot{x}(\tau)\) of the projected curve~\({\tau \mapsto x(\tau)}\) on \(M\) agrees with the temporal component \(f_0(\tau)\). Curves of this type are canonical lifts of trajectories from \(M\) to \(O\) and describe the time evolution of observers on Finsler spacetimes. The second condition, together with the defining property~\eqref{eqn:covderivative2} of \(\omega\), implies that the frame \(f\) is parallely transported along the projected curve \(\tau \mapsto (x(\tau), f_0(\tau))\) on \(O\). The flow of \(\mathbf{t}\) therefore can be interpreted as the time evolution of freely falling observer frames.

Note further that the projected curves \(\tau \mapsto (x(\tau), f_0(\tau))\) on \(O\) that satisfy~\eqref{eqn:canonicallift} are exactly the integral curves of the Reeb vector field \(\mathbf{r}\) defined in section~\ref{subsec:observerspace}. This reflects the fact that \(\mathbf{r}\) is the projection of \(\mathbf{t}\) along \(\pi\), i.e.,
\begin{equation}\label{eqn:timefields}
\mathbf{r} \circ \pi = \pi_* \circ \mathbf{t}\,.
\end{equation}
This will become important when we construct a Finsler spacetime out of a Cartan geometry in section~\ref{sec:cartanfinsler}.

\subsection{Curvature of the Cartan connection}\label{subsec:curvature}
In this final section we will complete our discussion of the Cartan geometry of Finsler observer spaces by calculating the curvature \(F \in \Omega^2(P,\mathfrak{g})\) of the Cartan connection, which measures the deviation of a Cartan geometry from its model Klein geometry and is given by
\begin{equation}
F = dA + \frac{1}{2}[A,A]\,.
\end{equation}
Again we make use of the split \(\mathfrak{g} = \mathfrak{h} \oplus \mathfrak{z}\) and decompose \(F\) into a $\mathfrak{h}$-valued part \(F_\mathfrak{h}\) and a $\mathfrak{z}$-valued part \(F_\mathfrak{z}\). Using this decomposition and the corresponding grading of \(\mathfrak{g}\), which follows from the algebra relations~\eqref{eqn:algebra}, we find
\begin{subequations}
\begin{align}
F_\mathfrak{h} &= d\omega + \frac{1}{2}[\omega,\omega] + \frac{1}{2}[e,e]\,,\label{eqn:curvatureh1}\\
F_\mathfrak{z} &= de + [\omega,e]\,.\label{eqn:curvaturez1}
\end{align}
\end{subequations}
Note the appearance of the term \([e,e]\) in the first equation, which explicitly depends on the choice of one of the groups \(G\) displayed in~\eqref{eqn:groups}, corresponding to the sign of the cosmological constant \(\Lambda\) on the model Klein geometry \(G/K\). This becomes apparent when we use~\eqref{eqn:component} to expand the curvature into component notation and obtain
\begin{subequations}
\begin{align}
F_{\mathfrak{h}}{}^j{}_i &= d\omega^j{}_i + \omega^j{}_k \wedge \omega^k{}_i + \sgn\Lambda\,e^j \wedge e^k\eta_{ki}\,,\label{eqn:curvatureh2}\\
F_{\mathfrak{z}}^i &= de^i + \omega^i{}_j \wedge e^j\,.\label{eqn:curvaturez2}
\end{align}
\end{subequations}
Inserting the Cartan connection~\eqref{eqn:connectionz} and~\eqref{eqn:connectionh} we finally find
\begin{subequations}
\begin{align}
F_{\mathfrak{h}}{}^j{}_i &= -\frac{1}{2}f^{-1}{}^j_df_i^c\left(R^d{}_{cab}dx^a \wedge dx^b + 2P^d{}_{cab}dx^a \wedge \delta f_0^b + S^d{}_{cab}\delta f_0^a \wedge \delta f_0^b\right)\nonumber\\
&\phantom{=} + f^{-1}{}^j_af^{-1}{}^k_b\eta_{ik}\sgn\Lambda\,dx^a \wedge dx^b\,,\label{eqn:curvatureh3}\\
F_{\mathfrak{z}}^i &= -f^{-1}{}^i_aC^a{}_{bc}dx^b \wedge \delta f_0^c\,,\label{eqn:curvaturez3}
\end{align}
\end{subequations}
with \(\delta f_0^a = df_0^a + N^a{}_bdx^b\). The coefficients \(R^d{}_{cab}\), \(P^d{}_{cab}\) and \(S^d{}_{cab}\) introduced here are the curvature coefficients of the Cartan linear connection, which are given by~\cite{Bucataru}
\begin{subequations}
\begin{align}
R^d{}_{cab} &= \delta_bF^d{}_{ca} - \delta_aF^d{}_{cb} + F^e{}_{ca}F^d{}_{eb} - F^e{}_{cb}F^d{}_{ea} + C^d{}_{ce}(\delta_bN^e{}_a - \delta_aN^e{}_b)\,,\\
P^d{}_{cab} &= \bar{\partial}_bF^d{}_{ca} - \delta_aC^d{}_{cb} + F^e{}_{ca}C^d{}_{eb} - C^e{}_{cb}F^d{}_{ea} + C^d{}_{ce}\bar{\partial}_bN^e{}_a\,,\\
S^d{}_{cab} &= \bar{\partial}_bC^d{}_{ca} - \bar{\partial}_aC^d{}_{cb} + C^e{}_{ca}C^d{}_{eb} - C^e{}_{cb}C^d{}_{ea}\,.
\end{align}
\end{subequations}
The curvature, in the sense of Cartan geometry, of the Cartan connection \(A\) hence resembles the curvature, in the sense of Finsler geometry, of the Cartan linear connection we used in the definition of \(\omega\). We further find a term in \(F_{\mathfrak{h}}\) which explicitly depends on the sign of the cosmological constant. This term compensates for the intrinsic curvature of the underlying model Klein geometry \(G/K\).

This concludes our discussion of the Cartan geometry of Finsler observer spaces. We have explicitly constructed the Cartan connection and calculated some of its properties. In the following section~\ref{sec:cartanfinsler} we will see which of these properties become important when we ask ourselves the converse question and aim to construct a Finsler geometry out of an arbitrary observer space Cartan geometry.

\section{Finsler geometry of Cartan observer space}\label{sec:cartanfinsler}
In the preceding section~\ref{sec:finslercartan} we have shown by explicit construction that a Finsler spacetime~\((M,F)\) gives rise to a Cartan geometry on its observer space. We will now consider the inverse direction, i.e., we start with an arbitrary observer space Cartan geometry \((\pi: P \to O, A)\) and aim to construct a Finsler spacetime. We will show by explicit construction that this is possible only if the observer space Cartan geometry satisfies certain additional conditions. Our construction will proceed in three steps. First, we will construct the spacetime manifold \(M\) in section~\ref{subsec:spacetime}. We will then show in section~\ref{subsec:observers} how observers can be related to tangent vectors of their trajectories on \(M\). Finally, we will construct the Finsler function \(F\) and hence also the Finsler metric \(g^F\) in section~\ref{subsec:finslermetric}.

\subsection{Spacetime}\label{subsec:spacetime}
Given an arbitrary observer space Cartan geometry there is no a priori notion of spacetime or the spatial and temporal location of an observer. Our construction hence must yield both a spacetime manifold \(M\) and a surjection \(\pi': O \to M\) assigning locations to observers. Constructing these two objects corresponds to constructing an equivalence relation on \(O\), with \(M\) being the set of equivalence classes and \(\pi'\) being the canonical projection. Recall that for the observer space of a Finsler spacetime as shown in section~\ref{subsec:observerspace} these equivalence classes are given by three-dimensional submanifolds \(S_x\) of observers being at the same location \(x \in M\). We wish to keep this structure also for more general Cartan geometries and decompose \(O\) into three-dimensional submanifolds. We therefore need a description of the sets \(S_x\) in terms of the Cartan connection only.

Finsler geometry allows a global definition of the submanifolds \(S_x\) as connected components of level sets of the Finsler function \(F\) (or, equivalently, the auxiliary function \(L\)). This definition is given by condition~\ref{finsler:timelike} on Finsler spacetimes. In a purely connection-based approach such as Cartan geometry we do not have this possibility, since the Cartan connection \(A\) describes only the local geometry of observer space. This means that we need to describe the geometry of the submanifolds \(S_x\) by their tangent spaces. Recall from section~\ref{subsec:observerspace} that the tangent space \(T_yS_x\) at~\({y \in S_x}\) agrees with the vertical tangent space \(V_{(x,y)}O \subset T_{(x,y)}O\). Conversely, the sets \(S_x\) are the integral submanifolds of the vertical distribution \(VO\). We have seen in section~\ref{subsec:tangbundsplit} that \(VO\) is the projection of the boost distribution \(BP\) on \(P\) along the bundle map \(\pi\), and is thus defined in terms of the Cartan connection only. However, for a general observer space Cartan geometry the vertical distribution \(VO\) may be non-integrable, i.e., it may not coincide with the tangent spaces of integral submanifolds. We therefore impose:

\begin{cond}\label{cond:integrable}
The vertical distribution \(VO\) on \(O\) must be integrable.
\end{cond}

Given an integrable distribution \(VO\) on \(O\), it follows from Frobenius' theorem that there exists a foliation \(\mathcal{F}\) of \(O\) such that the leaves of \(\mathcal{F}\) are integral submanifolds of \(VO\). We thus define~\(M\) to be the leaf space of \(\mathcal{F}\) and \(\pi'\) as the canonical projection onto \(M\). If \(O\) is the observer space of a Finsler spacetime, this indeed yields the underlying spacetime manifold. However, while Frobenius' theorem guarantees the existence of a foliation \(\mathcal{F}\), it does not make any statements about the structure of its leaf space \(M\). In particular, it does not guarantee that \(M\) admits the structure of a smooth manifold if we start from an arbitrary observer space with integrable vertical distribution. Since the smooth manifold structure is a necessary ingredient for the interpretation of \(M\) as spacetime, we further impose:

\begin{cond}\label{cond:strictlysimple}
The foliation \(\mathcal{F}\) associated to the vertical distribution \(VO\) must be strictly simple.
\end{cond}

A foliation is strictly simple if and only if its leaf space \(M\) is a smooth Hausdorff manifold and the projection \(\pi'\) is a surjective submersion. The latter in particular implies that for each \(o \in O\) the differential \(\pi'_{*o}: T_oO \to T_{\pi'(o)}M\) is a surjection. Since the leaves of \(\mathcal{F}\) are integral submanifolds of \(VO\) it follows further that \(\ker\pi'_{*o} = V_oO\). From the decomposition \(T_oO = V_oO \oplus H_oO\) it then follows that \(\pi'_{*o}\) bijectively maps \(H_oO\) to \(T_{\pi'(o)}M\). This will become relevant for our constructions of observer trajectories and the Finsler metric in the following two sections.

\subsection{Observer trajectories}\label{subsec:observers}
On a Finsler spacetime \((M,F)\) observers follow trajectories \(\tau \mapsto x(\tau)\) whose tangent vectors~\({\dot{x}(\tau)}\) at each point are future unit timelike. Hence observer trajectories can canonically be lifted to trajectories \(\tau \mapsto (x(\tau),\dot{x}(\tau))\) on observer space \(O\). Conversely, trajectories \(\tau \mapsto o(\tau)\) on \(O\) can be projected to trajectories \(\tau \mapsto \pi'(o(\tau))\) on \(M\). The projection from \(O\) to \(M\) can obviously be generalized to arbitrary observer spaces that satisfy the conditions derived in the previous section, and thus admit the construction of an underlying spacetime together with the projection \(\pi'\). In this section we discuss whether also the canonical lift from \(M\) to \(O\) can be generalized. This will provide us with a relation between observers \(o \in O\) and their four-velocities on \(M\).

We have seen in section~\ref{subsec:observerspace} that the observer space of a Finsler spacetime is a seven-dimensional submanifold of its tangent space. Our reconstruction of a Finsler spacetime from an observer space therefore needs to provide us with an embedding \(\sigma: O \to TM\). In the preceding section the construction of the spacetime manifold already yielded a projection \(\pi': O \to M\) that assigns a spacetime point, or location, to each observer. Compatibility of the embedding \(\sigma\) with the projection \(\pi'\) requires that \(\sigma(o) \in T_{\pi'(o)}M\) for all \(o \in O\), i.e., \(\sigma(o)\) must be a tangent vector at the observer location \(\pi'(o)\).

Let \(\tau \mapsto o(\tau)\) be an observer trajectory on \(O\). Its embedding \(\tau \mapsto \sigma(o(\tau))\) into \(TM\) is a canonical lift of its projection \(\tau \mapsto \pi'(o(\tau))\) if and only if
\begin{equation}\label{eqn:liftcondition}
\sigma(o(\tau)) = \frac{d}{d\tau}\pi'(o(\tau)) = \pi'_*(\dot{o}(\tau))\,.
\end{equation}
In section~\ref{subsec:timetranslation} we have discussed a particular class of trajectories on \(O\) which are given by the integral curves of the Reeb vector field \(\mathbf{r}\) and correspond to freely falling observers. We have seen that in case of a Finsler spacetime these are precisely the canonical lifts of timelike Finsler geodesics on \(M\). The requirement that the integral curves of \(\mathbf{r}\) are canonical lifts together with~\eqref{eqn:liftcondition} uniquely determines the map \(\sigma\) to be
\begin{equation}\label{eqn:obsembedding}
\sigma = \pi'_* \circ \mathbf{r}\,.
\end{equation}
However, for a general observer space Cartan geometry this map \(\sigma\) may not be an embedding. We therefore have to impose as a further condition:

\begin{cond}\label{cond:embedding}
The map \(\sigma: O \to TM\) must be an embedding.
\end{cond}

Using our construction of the map \(\sigma\) shown above we can now answer the question posed at the beginning of this section which trajectories on \(O\) correspond to canonical lifts of trajectories on the underlying spacetime manifold \(M\). From~\eqref{eqn:liftcondition} and~\eqref{eqn:obsembedding} it follows that their tangent vectors~\({\dot{o}(\tau)}\) must satisfy
\begin{equation}
\pi'_*(\dot{o}(\tau) - \mathbf{r}(o(\tau))) = 0\,.
\end{equation}
In the preceding section we have seen that the kernel of \(\pi'_*\) is given by the vertical tangent bundle~\(VO\). Thus, an observer trajectory is a canonical lift if and only if the horizontal part of its tangent vectors is given by the Reeb vector field \(\mathbf{r}\). Note that this property is defined only by the split of the tangent bundle \(TO\) into horizontal and vertical subbundles and the Reeb vector field, both of which exist on arbitrary observer space Cartan geometries, even if there is no underlying spacetime.

\subsection{The Finsler metric}\label{subsec:finslermetric}
In the previous two sections we have constructed a spacetime manifold \(M\) and an embedding~\({\sigma: O \to TM}\) that allows an identification of a certain class of observer trajectories on \(O\) with canonical lifts of trajectories from \(M\) to \(TM\). In the final part of our construction displayed in this section we aim to interpret the tangent vectors \(\sigma(o)\) as four-velocities of observers and to construct orthonormal observer frames on \(M\) whose time component agrees with \(\sigma(o)\). For this purpose we need to find a suitable notion of orthonormality on the tangent bundle \(TM\) with respect to which the tangent vectors to observer curves are future unit timelike. This notion will be provided by a Finsler metric \(g^F\), which can be derived from a Finsler function \(F\). In the following we will construct the Finsler function and show that we must impose further conditions on the initial observer space Cartan geometry.

Recall from property~\ref{finsler:fhomorev} in section~\ref{subsec:finsler} that the Finsler function \(F: TM \to \mathbb{R}\) is a positive, reversible, homogeneous function of degree one in the fiber coordinates. In section~\ref{subsec:observerspace} we have further seen that it takes the value \(F = 1\) on unit timelike vectors, and thus in particular on the observer space of a Finsler spacetime. This has a few consequences for the reconstruction of a Finsler function from a more general Cartan geometry. In the preceding section we have constructed an embedding \(\sigma: O \to TM\) and identified the corresponding embedded submanifold \(\sigma(O)\) with the observer space of \(M\). It thus follows that the Finsler function on \(TM\) must satisfy \(F(\lambda\sigma(o)) = |\lambda|\) for all \(o \in O\) and \(\lambda \in \mathbb{R}\). This is possible only if \(\sigma(O)\) intersects each ray \(\mathbb{R}^+(x,y) = \{(x,\lambda y), \lambda > 0\}\) at most once. Further, since we wish to identify observers with future oriented vectors, \(\sigma(O)\) must not intersect both a ray and its opposite. We therefore impose:

\begin{cond}\label{cond:intersections}
Each line \(\mathbb{R}(x,y) = \{(x,\lambda y), \lambda \in \mathbb{R}\}\) must intersect the embedded submanifold \(\sigma(O)\) of \(TM\) at most once.
\end{cond}

This condition guarantees the existence of both a Finsler function and a corresponding Finsler metric \(g^F_{ab} = \bar{\partial}_a\bar{\partial}_bF^2/2\) on observer space and on the double cone \(\mathbb{R}\sigma(O)\) obtained from rescaling~\(\sigma(O)\). We do not obtain any information about parts of \(TM\) that lie outside of this double cone, especially not about spacelike or lightlike vectors. This should not be too surprising since the Cartan geometry we started from only describes the geometry of observer space.

Our construction does not guarantee that the Finsler metric is non-degenerate and of Lorentz signature. The signature may even be different in different regions of observer space and degenerate on the boundaries between them. To exclude these cases we finally impose:

\begin{cond}\label{cond:signature}
The Finsler metric \(g^F\) associated to the Finsler function constructed above is non-degenerate and of Lorentz signature on all of observer space.
\end{cond}

Summarizing our findings, we have constructed a smooth spacetime manifold \(M\), identified its observer space with an embedded submanifold \(\sigma(O) \subset TM\) and equipped the double cone \(\mathbb{R}\sigma(O)\) with a Finsler function \(F\), such that the Finsler metric \(g^F\) has Lorentz signature and the observer four-velocities in \(\sigma(O)\) are future unit timelike. We have thus obtained a Finsler geometry on \(M\), as far as it can be measured by physical observers following timelike trajectories. Our construction does not yield a Finsler function on the spacelike vectors in \(TM \setminus \sigma(O)\) or the null structure of~\({(M,F)}\). This stems from the fact that the initial Cartan geometry \((\pi: P \to O,A)\) describes only the geometry of observer space and does not contain any information on non-timelike vectors.

Together with our results from section~\ref{sec:finslercartan} we now have prescriptions to transform a Finsler geometry to a Cartan geometry on its observer space and vice versa. This leads to the question whether these two prescriptions are compatible in the sense that applying one after the other yields back the initial geometry. This will be discussed in the following section.

\section{Closing the circle}\label{sec:reconstruction}
In the previous two sections we have developed two constructions that allow us to derive an observer space Cartan geometry from a Finsler spacetime and vice versa, provided that certain conditions are satisfied. In this section we will chain these two constructions together and discuss whether they are compatible, i.e., whether one of them is the inverse of the other. This would establish a correspondence between Finsler spacetimes and a class of observer space Cartan geometries. We will approach this question from two sides. In section~\ref{subsec:recfinsler} we will start from a Finsler spacetime, calculate its observer space Cartan geometry as shown in section~\ref{sec:finslercartan} and aim to reconstruct the original Finsler geometry using the method from section~\ref{sec:cartanfinsler}. The converse way will be discussed in section~\ref{subsec:reccartan}. We start with an observer space Cartan geometry satisfying the conditions derived in section~\ref{sec:cartanfinsler}, derive a Finsler spacetime and show whether its observer space Cartan geometry constructed using the steps from section~\ref{sec:finslercartan} agrees with the one we started from.

\subsection{Reconstruction of a given Finsler spacetime}\label{subsec:recfinsler}
In section~\ref{sec:cartanfinsler} we have developed a procedure that allows us to derive a Finsler spacetime from an observer space Cartan geometry \((\pi: P \to O,A)\), provided that the latter satisfies a few simple conditions. We have not given any examples for Cartan geometries yet that satisfy these conditions. However, it seems natural to assume that they are satisfied for Cartan geometries that arise as the observer space of a given Finsler spacetime \((M,F)\) as displayed in section~\ref{sec:finslercartan}. In this section we will show that this is indeed the case, so that we can follow the steps detailed in section~\ref{sec:cartanfinsler} to derive another Finsler spacetime \((\hat{M},\hat{F})\). It will turn out that we can partially reconstruct the original Finsler spacetime \((M,F)\) we started from.

The construction displayed in this section consists of two parts. We show that the observer space Cartan geometry \((\pi: P \to O,A)\) obtained from a Finsler spacetime \((M,F)\) satisfies the conditions~\ref{cond:integrable} to~\ref{cond:signature} derived in section~\ref{sec:cartanfinsler}. From this it follows that we can construct a Finsler spacetime \((\hat{M},\hat{F})\) together with an embedding \(\sigma: O \to T\hat{M}\) of the observer space into \(T\hat{M}\). We further show the existence of a diffeomorphism \(\mu: M \to \hat{M}\) which preserves the observer space and its Finsler geometry, i.e., which satisfies
\begin{equation}
\mu_*(o) = \sigma(o)\,, \quad \hat{F}(\lambda\mu_*(o)) = F(\lambda o) = |\lambda|
\end{equation}
for all \(\lambda \in \mathbb{R}\) and \(o \in O\).

Let \((M,F)\) be a Finsler spacetime and \((\pi: P \to O,A)\) its observer space Cartan geometry. We start by following the steps shown in section~\ref{subsec:spacetime} to construct the spacetime manifold \(\hat{M}\). For this purpose we first need to show that condition~\ref{cond:integrable} is satisfied, i.e., that the vertical distribution \(VO\) on \(O\) is integrable. This can immediately be seen from the fact that the vertical tangent spaces on \(O\) agree with the tangent spaces to the submanifolds \(S_x\) with \(x \in M\) defined in property~\ref{finsler:timelike} of Finsler spacetimes. Thus, \(O\) is naturally foliated by the leaves \(S_x\). We denote the corresponding leaf space by
\begin{equation}
\hat{M} = \{S_x, x \in M\}\,.
\end{equation}
The canonical projection will be denoted by \(\hat{\pi}': O \to \hat{M}\). One can easily read off that there is a bijective map \(\mu: M \to \hat{M}\), which we use to carry the smooth Hausdorff manifold structure from~\(M\) to \(\hat{M}\). This turns \(\mu\) into a diffeomorphism. Note that \(\hat{\pi}' = \mu \circ \pi'\), where \(\pi': O \to M\) is the usual projection. The latter is a surjective submersion since its differential \(\pi'_*\) bijectively maps each horizontal tangent space \(H_oO\) for \(o \in O\) to the corresponding tangent space \(T_{\pi'(o)}M\). It thus follows that also \(\hat{\pi}'\) is a surjective submersion. This proves that condition~\ref{cond:strictlysimple} is satisfied.

In order to show that also the remaining conditions~\ref{cond:embedding} to~\ref{cond:signature} are satisfied, consider the following diagram:
\begin{equation}\label{eqn:recfinsler}
\xymatrix{TM \ar[dd]_{\mu_*} \ar[rrr] & & & M \ar[dd]^{\mu}\\
& TO \ar[ul]^{\pi'_*} \ar[dl]_{\hat{\pi}'_*} & O \ar[l]^{\mathbf{r}} \ar@{_(->}[ull] \ar[dll]^{\sigma} \ar[ur]_{\pi'} \ar[dr]^{\hat{\pi}'} &\\
T\hat{M} \ar[rrr] & & & \hat{M}}
\end{equation}
It is not difficult to check that this diagram commutes. The triangle on the right hand side commutes by the construction of the map \(\mu\). Differentiating this triangle yields the triangle on the left hand side, which therefore also commutes. The small triangle below commutes by the definition~\eqref{eqn:obsembedding} of the map \(\sigma\). From the definition~\eqref{eqn:reebvectorfield} of the Reeb vector field it finally follows that also the upper small triangle commutes, where \(O \hookrightarrow TM\) denotes the canonical embedding following from the construction of the observer space of a Finsler spacetime detailed in section~\ref{subsec:observerspace}.

Since \(\mu\) is a diffeomorphism from \(M\) to \(\hat{M}\), its differential \(\mu_*\) maps the embedded submanifold~\({O \subset TM}\) to an embedded submanifold of \(T\hat{M}\). From the commutativity of the triangle with corners \(O\), \(TM\) and \(T\hat{M}\) it then follows that \(\sigma\) is an embedding, which proves condition~\ref{cond:embedding}. Since~\(\mu_*\) is a vector bundle isomorphism, and thus preserves the lines \(\mathbb{R}o\) for \(o \in TM\), it follows that also condition~\ref{cond:intersections} is satisfied. We can thus apply the construction from section~\ref{subsec:finslermetric} and equip the space~\({\mathbb{R}\sigma(O) \subset T\hat{M}}\) with a Finsler function \(\hat{F}\). From the commutativity of~\eqref{eqn:recfinsler} it further follows that \(F = \hat{F} \circ \mu_*\) on \(\mathbb{R}O \subset TM\), i.e., the Finsler function is preserved by the map \(\mu_*\). Hence, also the Finsler metric and in particular its signature are preserved, which proves condition~\ref{cond:signature}.

Our findings can be summarized as follows. Starting from a Finsler spacetime \((M,F)\) and applying the constructions from sections~\ref{sec:finslercartan} and~\ref{sec:cartanfinsler}, we obtain a spacetime manifold \(\hat{M}\) which is diffeomorphic to \(M\). The differential of the diffeomorphism \(\mu\) maps the observer space \(O \in TM\) of \(M\) to the observer space of \(\hat{M}\). On the observer space of \(\hat{M}\) we obtain a Finsler function, and hence a Finsler metric, which agrees with the Finsler metric on the observer space of \(M\). We have thus reconstructed the initial Finsler geometry on the timelike vectors of \(TM\), which correspond to physical observers following timelike trajectories. As already remarked in section~\ref{sec:cartanfinsler}, we cannot reconstruct the Finsler function on spacelike vectors of \(TM\) or the null structure of \((M,F)\), since these do not enter the construction of the observer space Cartan geometry.

We conclude this section with a remark on the reconstruction of the Finsler metric. A well-known result from Cartan geometry states that a Riemannian manifold \((M,g)\) gives rise to a Cartan geometry on the model \(G/H\) with \(G = \mathrm{ISO}_0(3,1)\) and \(H = \mathrm{SO}_0(3,1)\), and that the Riemannian metric \(g\) can be reconstructed from the Cartan geometry only up to a global scale factor~\cite{Sharpe}. This freedom to choose a scale factor corresponds to the rescaling freedom of the generators \(\mathcal{Z}_i\) of translation, which does not change the algebra relations~\eqref{eqn:algebra} for \(\Lambda = 0\). The same rescaling freedom applies to our construction shown above. The choice of a different scale factor leads to a rescaling of the Reeb vector field \(\mathbf{r}\), the embedding \(\sigma\), the Finsler function \(\hat{F}\) and thus finally the Finsler metric \(\hat{g}^F\). Note that for \(\Lambda \neq 0\) the scale factor is fixed by the algebra relations~\eqref{eqn:algebra}.

\subsection{Reconstruction of a given Cartan observer space}\label{subsec:reccartan}
In the preceding section we have seen that if we start from a Finsler spacetime and construct its observer space Cartan geometry, we can partially reconstruct the initial Finsler geometry. We now discuss the inverse problem and start with an observer space Cartan geometry~\({(\pi: P \to O,A)}\) that satisfies conditions~\ref{cond:integrable} to~\ref{cond:signature} from section~\ref{sec:cartanfinsler}. Following the construction displayed in the same section we derive a Finsler geometry \((M,F)\). We then construct its observer space Cartan geometry~\({(\hat{\pi}: \hat{P} \to \hat{O},\hat{A})}\) as shown in section~\ref{sec:finslercartan} and pose the question in which cases we re-obtain the initial Cartan geometry. We will answer this question by deriving a simple condition on the initial Cartan geometry, which allows us to identify the Cartan geometries which can arise as observer space geometries of Finsler spacetimes.

In order to show that the Cartan geometries \((\pi: P \to O,A)\) and \((\hat{\pi}: \hat{P} \to \hat{O},\hat{A})\) agree, we need to find a diffeomorphism \(\chi: P \to \hat{P}\) that satisfies the following two conditions:
\begin{enumerate}[(i)]
\item
The diagram
\begin{equation}\label{eqn:bundleiso}
\xymatrix{P \ar[d]_{\chi} \ar[r]^{\pi} & O \ar[d]_{\sigma} \ar[r]^{\pi'} & M\\
\hat{P} \ar[r]_{\hat{\pi}} & \hat{O} \ar[ur]_{\hat{\pi}'}}
\end{equation}
commutes. The square on the left hand side of the diagram simply states that \(\chi\) is a bundle morphism between two principal $K$-bundles. The triangle on the right hand side commutes due to the construction of the embedding \(\sigma\) and is displayed here for later use.
\item
The Cartan connection \(A\) agrees with the pullback of the Cartan connection \(\hat{A}\) along \(\chi\). This condition can equivalently be formulated in terms of the fundamental vector fields: for each~\({a \in \mathfrak{g}}\), the vector field \(\underline{\hat{A}}(a)\) on \(\hat{P}\) is the pushforward of the corresponding vector field~\(\underline{A}(a)\) on \(P\) along \(\chi\), i.e., the diagram
\begin{equation}\label{eqn:fundvectpush}
\xymatrix{P \ar[d]_{\chi} \ar[r]^{\underline{A}(a)} & TP \ar[d]^{\chi_*}\\
\hat{P} \ar[r]_{\underline{\hat{A}}(a)} & T\hat{P}}
\end{equation}
commutes.
\end{enumerate}
We will call a map satisfying these conditions a \emph{Cartan morphism}. In order to prove the existence or non-existence a Cartan morphism, the following proposition will be helpful:

\begin{prop}
If a Cartan morphism \(\chi = (x,f): P \to \hat{P}\) for a given initial Cartan geometry~\({(\pi: P \to O,A)}\) and the reconstructed Cartan geometry \((\hat{\pi}: \hat{P} \to \hat{O},\hat{A})\) exists, its constituent functions \(x: P \to M\) and \(f_i: P \to TM\) are given by
\begin{equation}\label{eqn:cartanmorphism}
x(p) = \pi'(\pi(p))\,, \quad f_i(p) = \pi'_*(\pi_*(\underline{A}(\mathcal{Z}_i)(p)))\,.
\end{equation}
\end{prop}

\begin{proof}
Since the elements of \(\hat{P}\) are frames on \(M\), we can decompose \(\chi\) into maps~\({x: P \to M}\) and~\({f_i: P \to TM}\) assigning to \(p \in P\) the base point \(x(p) \in M\) and the constituent vectors~\({f_i(p) \in T_{x(p)}M}\) of the frame \(\chi(p)\). The base point \(x(p)\) is given by the canonical projection~\({\hat{\pi}' \circ \hat{\pi}: \hat{P} \to M}\) and can thus be written as
\begin{equation}
x(p) = \hat{\pi}'(\hat{\pi}(\chi(p))) = \pi'(\pi(p))\,,
\end{equation}
where the last equality follows from the commutativity of the diagram~\eqref{eqn:bundleiso}. Further, consider the diagram
\begin{equation}
\xymatrix{P \ar[d]_{\chi} \ar[r]^{\underline{A}(a)} & TP \ar[d]_{\chi_*} \ar[r]^{\pi_*} & TO \ar[d]_{\sigma_*} \ar[r]^{\pi'_*} & TM\\
\hat{P} \ar[r]_{\underline{\hat{A}}(a)} & T\hat{P} \ar[r]_{\hat{\pi}_*} & T\hat{O} \ar[ur]_{\hat{\pi}'_*}}
\end{equation}
for \(a = \mathcal{Z}_i \in \mathfrak{g}\). The square on the left commutes due to the commutativity of~\eqref{eqn:fundvectpush}, while the remainder of the diagram is simply the differential of~\eqref{eqn:bundleiso} and therefore also commutes. From this we find
\begin{equation}
f_i(p) = \hat{\pi}'_*(\hat{\pi}_*(\underline{\hat{A}}(\mathcal{Z}_i)(\chi(p)))) = \pi'_*(\pi_*(\underline{A}(\mathcal{Z}_i)(p)))\,,
\end{equation}
where the first equality holds due to the relation~\eqref{eqn:fundvecfields} between the frame components \(f_i(p)\) and the fundamental vector fields \(\underline{\hat{A}}(\mathcal{Z}_i)(\chi(p))\). This completes the proof.
\end{proof}

This proposition allows us to prove the existence or non-existence of a Cartan morphism \(\chi\), and thus the equivalence or non-equivalence of the initial and reconstructed Cartan geometries, by simply examining the properties of the maps~\eqref{eqn:cartanmorphism}. If these maps constitute orthonormal frames and the combined map \(\chi = (x,f)\) is a Cartan morphism, we have shown its existence by explicit construction. If the maps~\eqref{eqn:cartanmorphism} fail to constitute an orthonormal frame, or the combined map \(\chi\) fails to be a Cartan morphism, the proposition implies the non-existence of a Cartan morphism.

This result concludes our discussion of the relation between Finsler spacetimes and observer space Cartan geometries. We have shown in the preceding section~\ref{subsec:recfinsler} that if we start from a Finsler spacetime and calculate its observer space Cartan geometry using the method detailed in section~\ref{sec:finslercartan}, we can partially reconstruct the initial Finsler geometry using the method detailed in section~\ref{sec:cartanfinsler}. In this section we have discussed the converse statement and derived under which conditions one re-obtains a given observer space Cartan geometry by applying the methods from sections~\ref{sec:cartanfinsler} and~\ref{sec:finslercartan} in the opposite order. In the remainder of this article we will apply our findings to gravity theories based on either of these two geometries.

\section{Application to gravity actions}\label{sec:gravity}
We have shown how to obtain a Cartan geometry on observer space from a Finsler spacetime in section~\ref{sec:finslercartan} and vice versa in section~\ref{sec:cartanfinsler}. Our discussion has so far been purely kinematic. We will now turn our focus to the gravitational dynamics of these two different geometrical models. For this purpose we will consider gravity actions on either of these geometries. First, we will discuss MacDowell-Mansouri gravity on observer space and translate the action into Finsler language in section~\ref{subsec:cartangrav}. We will then follow the opposite direction and translate Finsler gravity to Cartan language in section~\ref{subsec:finslergrav}.

\subsection{MacDowell-Mansouri gravity in Finsler language}\label{subsec:cartangrav}
The first gravity theory we discuss in this article is a lift of MacDowell-Mansouri gravity to observer space~\cite{Gielen:2012fz}. The essential part of the gravitational action can be written as
\begin{equation}\label{eqn:cartangrav}
S_G = \int_O\kappa_{\mathfrak{h}}(F_{\mathfrak{h}} \wedge F_{\mathfrak{h}}) \wedge \tau_{\mathfrak{y}}(b \wedge b \wedge b)\,.
\end{equation}
Note that this is an action on \(O\), although the terms in the action are forms which are defined on \(P\). It is to be understood that these forms are pulled back along a section of the observer bundle~\(\pi: P \to O\). The action is chosen so that the pullback is independent of the choice of this section. We will now explain the terms in the action in detail.

We start with the last term \(\tau_{\mathfrak{y}}(b \wedge b \wedge b)\), where \(b\) denotes the $\mathfrak{y}$-valued part of the Cartan connection \(A\) and \(\tau_{\mathfrak{y}}\) is a $K$-invariant trilinear form on \(\mathfrak{y}\). For simplicity, we choose this to be
\begin{equation}
\tau_{\mathfrak{y}}(b \wedge b \wedge b) = \epsilon_{\alpha\beta\gamma}b^{\alpha} \wedge b^{\beta} \wedge b^{\gamma}\,.
\end{equation}
Similarly, \(\kappa_{\mathfrak{h}}\) in the first term denotes a non-degenerate, $H$-invariant inner product on \(\mathfrak{h}\). Again we make a simple choice and set
\begin{equation}
\kappa_{\mathfrak{h}}(h,h') = \tr_{\mathfrak{h}}(h,\star h')\,,
\end{equation}
for \(h,h' \in \mathfrak{h}\), where \(\tr_{\mathfrak{h}}\) denotes the Killing form on \(\mathfrak{h}\) and \(\star\) is a Hodge star operator. We remark that other choices are possible and related to the Immirzi parameter~\cite{Gielen:2012fz,Wise:2009fu}. In components we can write the Killing form as
\begin{equation}
\tr_{\mathfrak{h}}(h,h') = h^i{}_jh'^j{}_i
\end{equation}
 and the Hodge star operator as
\begin{equation}
(\star h)^i{}_j = \eta^{im}\eta^{ln}\epsilon_{mjkl}h^k{}_n\,.
\end{equation}
Finally, \(F_{\mathfrak{h}}\) denotes the $\mathfrak{h}$-valued part~\eqref{eqn:curvatureh1} of the curvature \(F\) of the Cartan connection. Recall that \(F_{\mathfrak{h}} = R + \frac{1}{2}[e,e]\) splits into the curvature \(R\) of \(\omega\) and a purely algebraic part \([e,e]\) which depends on the choice of the group \(G\) in the model Klein geometry \(G/H\). For MacDowell-Mansouri gravity we are forced to choose \(G\) so that \([e,e] \neq 0\), i.e., we must choose a model Klein geometry which corresponds to a non-vanishing cosmological constant. The reason for this restriction becomes apparent below. We can use the split of \(F_{\mathfrak{h}}\) to decompose \(\tr_{\mathfrak{h}}(F_{\mathfrak{h}} \wedge \star F_{\mathfrak{h}})\) into the following components:
\begin{itemize}
\item
A cosmological constant term:
\begin{equation}
\frac{1}{4}\tr_{\mathfrak{h}}([e,e] \wedge \star[e,e]) = -(\sgn\Lambda)^2\epsilon_{ijkl}e^i \wedge e^j \wedge e^k \wedge e^l\,.
\end{equation}
\item
A curvature term:
\begin{equation}
\tr_{\mathfrak{h}}([e,e] \wedge \star R) = \frac{1}{6}\sgn\Lambda\,g^{F\,ab}R^c{}_{acb}\epsilon_{ijkl}e^i \wedge e^j \wedge e^k \wedge e^l + (\ldots)\,.
\end{equation}
\item
A Gauss-Bonnet term:
\begin{equation}
\tr_{\mathfrak{h}}(R \wedge \star R) = -\frac{1}{96}R_{abcd}R_{efgh}\epsilon^{abef}\epsilon^{cdgh}\epsilon_{ijkl}e^i \wedge e^j \wedge e^k \wedge e^l + (\ldots)\,.
\end{equation}
\end{itemize}
The ellipsis in the expressions above indicates that we have omitted terms which are of the form \(c \wedge b\), where \(c\) is a 3-form and \(b\) is the boost part of the Cartan connection. These terms do not contribute to the total action since their wedge product with \(\tau_{\mathfrak{y}}(b \wedge b \wedge b)\) vanishes. Note the appearance of the 7-form
\begin{equation}\label{eqn:cartanvolume}
\Omega_S = \epsilon_{ijkl}\epsilon_{\alpha\beta\gamma}e^i \wedge e^j \wedge e^k \wedge e^l \wedge b^{\alpha} \wedge b^{\beta} \wedge b^{\gamma}
\end{equation}
in all terms of the action \(S_G\). One easily checks that this is the volume form of the pullback
\begin{equation}
\tilde{G}_S = \eta_{ij}f^{-1}{}^i_af^{-1}{}^j_b\,dx^a \otimes dx^b + \delta_{\alpha\beta}f^{-1}{}^{\alpha}_af^{-1}{}^{\beta}_b\,\delta y^a \otimes \delta y^b = \eta_{ij}e^i \otimes e^j + \delta_{\alpha\beta}b^{\alpha} \otimes b^{\beta}
\end{equation}
of the Sasaki metric~\eqref{eqn:sasakimetric} from \(TM\) to \(O\). The total action then takes the form
\begin{equation}
S_G = \int_O\left(\frac{1}{6}\sgn\Lambda\,g^{F\,ab}R^c{}_{acb} - \frac{1}{96}R_{abcd}R_{efgh}\epsilon^{abef}\epsilon^{cdgh} - (\sgn\Lambda)^2\right)\Omega_S
\end{equation}
and is written in terms of Finsler geometry only. From this we see that in the case \(\Lambda = 0\), which corresponds to the choice \(G = \mathrm{ISO}_0(3,1)\), both the cosmological constant term and the curvature term vanish, so that only the Gauss-Bonnet term remains. A non-vanishing cosmological constant is therefore required to obtain an action which contains a term equivalent to the Ricci scalar, and thus resembles the Einstein-Hilbert action in the limit of a Lorentzian spacetime. Note that the relative weight of the terms in the resulting action, and thus the magnitude of the cosmological constant, can be adjusted by introducing appropriate scale factors into the algebra relations~\eqref{eqn:algebra} and the choice of basis vectors \(\mathcal{Z}_i\).

\subsection{Finsler gravity in Cartan language}\label{subsec:finslergrav}
In the preceding section we have shown how to rewrite the action of MacDowell-Mansouri gravity on observer space, given in Cartan language, using objects of Finsler geometry. We will now follow the opposite direction and translate a Finsler geometric gravity action into Cartan language. For this purpose we start from the gravity action~\cite{Pfeifer:2011xi}
\begin{equation}\label{eqn:finslergrav}
S_G = \int_OR^a{}_{ab}y^b\Omega_S\,,
\end{equation}
where \(R^a{}_{bc}\) denotes the curvature of the non-linear connection given by
\begin{equation}\label{eqn:nonlincurv}
R^a{}_{bc} = y^d(\delta_cF^a{}_{bd} - \delta_bF^a{}_{cd} + F^e{}_{bd}F^a{}_{ce} - F^e{}_{cd}F^a{}_{be})\,,
\end{equation}
and \(\Omega_S\) is the volume form of the pullback of the Sasaki metric~\eqref{eqn:sasakimetric} from \(TM\) to \(O\). We have already seen in~\eqref{eqn:cartanvolume} how to rewrite the volume form \(\Omega_S\) in terms of the Cartan connection. In order to translate the curvature tensor~\eqref{eqn:nonlincurv} into Cartan language we use the following geometric interpretation. The curvature \(R^a{}_{bc}\) of the non-linear connection measures the non-integrability of the horizontal tangent bundle \(HTM \subset TTM\) and can be calculated from the Lie bracket~\cite{Bucataru}
\begin{equation}
[\delta_b,\delta_c] = R^a{}_{bc}\bar{\partial}_a
\end{equation}
of the horizontal vector fields \(\delta_a = \partial_a - N^b{}_a\bar{\partial}_b\) that span \(HTM\). It thus appears natural to define the corresponding curvature on observer space as the vertical part of the Lie bracket of horizontal vector fields on \(O\). However, it turns out to be easier to work on \(P\) instead of \(O\). Recall that the horizontal tangent bundle \(HP\) is spanned by the fundamental vector fields~\eqref{eqn:vecfieldz}
\begin{equation}
\underline{A}(\mathcal{Z}_i) = f_i^a(\partial_a - f_j^bF^c{}_{ab}\bar{\partial}^j_c)\,,
\end{equation}
where \(\mathcal{Z}_i \in \mathfrak{z}\) are the generators of translations. Their Lie bracket yields
\begin{equation}
[\underline{A}(\mathcal{Z}_i),\underline{A}(\mathcal{Z}_j)] = f_i^bf_j^cf_k^d(\delta_cF^a{}_{bd} - \delta_bF^a{}_{cd} + F^e{}_{bd}F^a{}_{ce} - F^e{}_{cd}F^a{}_{be})\bar{\partial}^k_a\,,
\end{equation}
which does indeed reproduce the components~\eqref{eqn:nonlincurv} of the curvature tensor \(R^a{}_{bc}\). In order to select only the components we are interested in, we now project the result to the boost part \(BP\) of the tangent bundle, i.e., we apply the 1-form \(b \in \Omega^1(P,\mathfrak{y})\) given by \(b^{\alpha} = \omega^{\alpha}{}_0\) with the Cartan connection~\eqref{eqn:connectionh}. From this we obtain
\begin{equation}
b^{\alpha}([\underline{A}(\mathcal{Z}_i),\underline{A}(\mathcal{Z}_j)]) = f^{-1}{}^{\alpha}_af_i^bf_j^cR^a{}_{bc}\,.
\end{equation}
To reproduce the term \(R^a{}_{ab}y^b\) in the Finsler gravity action~\eqref{eqn:finslergrav}, we finally calculate the contraction
\begin{equation}
b^{\alpha}([\underline{A}(\mathcal{Z}_{\alpha}),\underline{A}(\mathcal{Z}_0)]) = f_0^bR^a{}_{ab}\,.
\end{equation}
Note that the result is a function on \(P\) which is invariant under the action of \(K\), i.e., it is constant along the fibers of the principal bundle \(\pi: P \to O\). This is a consequence of the equivariance of the Cartan connection. It therefore descends to a function on \(O\). Combining this with the result with the expression~\eqref{eqn:cartanvolume} for the volume form \(\Omega_S\) we finally arrive at the action
\begin{equation}\label{eqn:cartanaction}
S_G = \int_Ob^{\alpha}([\underline{A}(\mathcal{Z}_{\alpha}),\underline{A}(\mathcal{Z}_0)])\Omega_S\,,
\end{equation}
which is now completely rewritten in terms of the Cartan connection and its fundamental vector fields. Note that since observer space Cartan geometry is more general than Finsler geometry, this translated action is not unique. We may add further terms to the action that vanish in case the Cartan geometry is induced by a Finsler geometry. We will not discuss terms of this type here, and leave the action~\eqref{eqn:cartanaction} as it is.

This concludes our discussion of gravity in the context of Finsler geometry and Cartan geometry on observer space. We have shown that we can translate the actions of gravity theories on observer space from Cartan language to Finsler language and vice versa. This proves the practical use of the constructions displayed in sections~\ref{sec:finslercartan} and~\ref{sec:cartanfinsler} in gravitational physics.

\section{Conclusion}\label{sec:conclusion}
In this article we have revealed a close connection between two seemingly different extensions to the Lorentzian geometry of spacetime: Cartan geometry of observer space and Finsler spacetimes. We have shown that each Finsler spacetime possesses a well-defined observer space which can be equipped with a Cartan geometry. Conversely, we have derived conditions under which an observer space Cartan geometry gives rise to a Finsler spacetime. We have further shown that these two constructions complement each other in the sense that one can partially reconstruct a given Finsler or Cartan geometry by applying both constructions in the appropriate order.

Our work provides the mathematical foundation for translating physical theories from Finsler to Cartan geometry and vice versa. We have demonstrated this possibility by applying our construction to two gravity theories: Finsler gravity and MacDowell-Mansouri gravity on observer space. It turned out that their actions can easily be translated between the two different geometric descriptions of spacetime.

It is our aim to further investigate the properties of the two translated gravity theories, their field equations and their solutions. Moreover, we aim to apply our construction to further physical theories on either side of the relation. While in this article we have restricted ourselves to gravity theories in vacuum, one may naturally pose the question whether one can also translate matter field theories from Finsler to Cartan geometry and vice versa. In particular it needs to be shown whether a consistent matter coupling to one of the discussed background geometries of spacetime translates into a consistent coupling to the other one.

Finally, our construction provides a possible starting point for the geometrodynamic description of Finsler geometry. The central concept of geometrodynamics is the time evolution of a spatial geometry, which requires a split of the geometry of spacetime into spatial and temporal components. The Cartan geometric description of observer space naturally provides this split, and hence gives rise to the notion of Cartan geometrodynamics on Lorentzian spacetimes~\cite{Gielen:2011mk}. The work presented in this article extends this framework to Finsler spacetimes. In particular, it provides a potential Finsler geometric generalization to the derivation of covariant Ashtekar variables from observer space Cartan geometry~\cite{Gielen:2012fz}. Further work is required to show how this is related to previous generalizations of Ashtekar variables to Finsler geometry~\cite{Vacaru:2008qf}, and to connect Finsler gravity to sum-over-histories formulations such as spin foam models or causal dynamical triangulations.

\acknowledgments
The author is happy to thank Steffen Gielen, Christian Pfeifer, Sergiu Vacaru and Derek Wise for helpful comments and discussions. He gratefully acknowledges full financial support from the Estonian Science Foundation through the research grant ERMOS115.

\end{document}